\documentclass[11pt,letterpaper]{article}

\usepackage{graphicx}
\usepackage[margin=1in]{geometry}
\usepackage{enumitem}
\usepackage{mleftright}
\usepackage{dsfont}
\usepackage{mathtools}
\usepackage{amsmath, amssymb, amsthm}
\usepackage{thmtools}
\usepackage{mathabx}
\usepackage{bm}
\usepackage{nicefrac}       
\usepackage{microtype}      

\usepackage{silence}
\usepackage{comment}
\usepackage{bbm}
\usepackage{aligned-overset}
\usepackage{ifthen}
\usepackage[colorinlistoftodos,textsize=scriptsize]{todonotes}
\usepackage{multicol}
\usepackage{booktabs}

\newcounter{constant}

\newcommand{\newconstant}[1]{%
  \begingroup
    \refstepcounter{constant}%
    \label{#1}%
  \endgroup
}

\newcommand{\useconstant}[1]{C_{\ref{#1}}}

\makeatletter
\def\namedlabel#1#2{\begingroup
    #2%
    \def\@currentlabel{#2}%
    \phantomsection\label{#1}\endgroup
}
\makeatother

\PassOptionsToPackage{hyphens}{url}
\usepackage[backref=page, linktocpage]{hyperref}
\hypersetup{
    colorlinks,
    citecolor=blue,
    filecolor=blue,
    linkcolor=blue,
    urlcolor=blue
}
\renewcommand*{\backref}[1]{}
\renewcommand*{\backrefalt}[4]{%
    \ifcase #1 %
    \or     #2%
    \else   #2%
    \fi
}
\usepackage{listings}
\usepackage{fancyhdr}

\usepackage[T1]{fontenc}
\usepackage{framed, color}
\usepackage{lmodern}
\usepackage{bigints}
\usepackage{cancel}
\allowdisplaybreaks
\DeclareSymbolFont{rmlargesymbols}{OMX}{mdbch}{m}{n}
\DeclareMathSymbol{\rmintop}{\mathop}{rmlargesymbols}{82}

\AtBeginDocument{%
  \DeclareFontShape{T1}{lmr}{m}{scit}{<->ssub*lmr/m/scsl}{}%
}

\usepackage{algorithm}
\usepackage{algpseudocode}

\let\originalleft\left
\let\originalright\right
\renewcommand{\left}{\mathopen{}\mathclose\bgroup\originalleft}
\renewcommand{\right}{\aftergroup\egroup\originalright}

\definecolor{codegreen}{rgb}{0,0.6,0}
\definecolor{codegray}{rgb}{0.5,0.5,0.5}
\definecolor{codepurple}{rgb}{0.58,0,0.82}
\definecolor{backcolour}{rgb}{0.95,0.95,0.92}

\lstdefinestyle{mystyle}{
    backgroundcolor=\color{backcolour},   
    commentstyle=\color{codegreen},
    keywordstyle=\color{magenta},
    numberstyle=\tiny\color{codegray},
    stringstyle=\color{codepurple},
    basicstyle=\footnotesize,
    breakatwhitespace=false,         
    breaklines=true,                 
    captionpos=b,                    
    keepspaces=true,                 
    numbers=left,                    
    numbersep=5pt,                  
    showspaces=false,                
    showstringspaces=false,
    showtabs=false,                  
    tabsize=2
}
 
\makeatletter
\newcommand{\AlgoResetCount}{\renewcommand{\@ResetCounterIfNeeded}{\setcounter{AlgoLine}{0}}}
\newcommand{\AlgoNoResetCount}{\renewcommand{\@ResetCounterIfNeeded}{}}
\newcounter{AlgoSavedLineCount}

\makeatother

\usepackage{tikz}
\tikzset{every picture/.style={line width=0.75pt}} 

\newcommand{\bR}{\mathbb{R}}

\newcommand{\bZ}{\mathbb{Z}}

\newcommand{\cA}{\mathcal{A}}

\newcommand{\cD}{\mathcal{D}}

\newcommand{\cH}{\mathcal{H}}

\newcommand{\cO}{\mathcal{O}}

\newcommand{\cX}{\mathcal{X}}
\newcommand{\cY}{\mathcal{Y}}
\newcommand{\cZ}{\mathcal{Z}}

\newcommand{\Bern}{\mathrm{Bern}}

\newcommand{\TV}{d_{\mathrm{TV}}}

\newcommand{\clip}{\mathtt{clip}}

\newcommand{\paren}[1]{\left( #1 \right)}
\newcommand{\brk}[1]{\left[ #1 \right]}
\newcommand{\brc}[1]{\left\{ #1 \right\}}

\newcommand{\pr}[2]{
    \ifthenelse{\equal{#1}{}}{
        \mathbb{P}\left[ #2 \right]
    }
    {
        \underset{#1}{\mathbb{P}}\left[ #2 \right]
    }
}
\newcommand{\ex}[2]{
    \ifthenelse{\equal{#1}{}}{
        \mathbb{E}\left[ #2 \right]
    }
    {
        \underset{#1}{\mathbb{E}}\left[ #2 \right]
    }
}
\newcommand{\var}[2]{
    \ifthenelse{\equal{#1}{}}{
        \mathrm{Var}\left( #2 \right)
    }
    {
        \underset{#1}{\mathrm{Var}}\left( #2 \right)
    }
}
\newcommand{\llnorm}[1]{\left\| #1 \right\|_2}

\newcommand{\mnorm}[2]{\left\| #2 \right\|_{#1}}

\newcommand{\abs}[1]{\left| #1 \right|}

\newcommand{\eps}{\varepsilon}

\definecolor{shadecolor}{rgb}{0.83, 0.83, 0.83}
\theoremstyle{plain}

\newtheorem{thm}{Theorem}[section]
\newtheorem{theorem}[thm]{Theorem}
\newtheorem{fact}[thm]{Fact}
\newtheorem{proposition}[thm]{Proposition}
\newtheorem{corollary}[thm]{Corollary}
\newtheorem{lemma}[thm]{Lemma}



\theoremstyle{definition}
\newtheorem{definition}[thm]{Definition}

\newtheorem{question}[thm]{Question}
\newtheorem{remark}[thm]{Remark}



\title{Not All Learnable Distribution Classes are Privately Learnable\thanks{Authors are listed in alphabetical order.}}
\author{
    Mark Bun\thanks{{\tt mbun@bu.edu}. Boston University.}
\and
    Gautam Kamath\thanks{{\tt g@csail.mit.edu}. Cheriton School of Computer Science, University of Waterloo and Vector Institute.}
\and
    Argyris Mouzakis\thanks{{\tt amouzaki@uwaterloo.ca}. Cheriton School of Computer Science, University of Waterloo.}
\and
    Vikrant Singhal\thanks{{\tt vikrant@seas.harvard.edu}. John A. Paulson School of Engineering and Applied Sciences, Harvard University. Work done while being a postdoc at the Cheriton School of Computer Science, University of Waterloo.}
}

\begin{document}

\maketitle
\begin{abstract}
We give an example of a class of distributions that is learnable up to constant error in total variation distance with a finite number of samples, but not learnable under $\paren{\eps, \delta}$-differential privacy with the same target error.
This weakly refutes a conjecture of Ashtiani.
\end{abstract}

\newpage

\tableofcontents

\newpage

\section{Introduction}
\label{sec:intro}

Given samples from a distribution $\cD$ belonging to some class of distributions $\cH$, can we output a distribution $\cD'$ that is close to $\cD$ in total variation distance?
This problem, known as \emph{distribution learning} or \emph{density estimation}, has enjoyed significant study by a number of communities, including Computer Science, Statistics, and Information Theory (see, e.g.,~\cite{DevroyeL01, KearnsMRRSS94, DaskalakisDS12b, AshtianiBHLMP20}). 

A recent line of work studies distribution learning under \emph{differential privacy}~\cite{DworkMNS06}, giving sample complexity bounds for several classes of interest.
However, many of these algorithms are ad hoc, exploiting idiosyncrasies of the class of interest (see, e.g.,~\cite{KarwaV18, KamathLSU19}). 
Recent efforts have succeeded in weakening assumptions and designing increasingly general learning algorithms and frameworks (see, e.g.,~\cite{LiuKO22, KamathMSSU22, AshtianiL22, KothariMV22, AfzaliAL23}).
It is natural to wonder how far this agenda can be pushed -- what are the limits of private learning?
Specifically, we consider the following question:

\begin{question}
\label{q:learnability}
Is every learnable class of distributions $\cH$ also learnable under the constraint of $\paren{\eps, \delta}$-differential privacy?
\end{question}

The answer is known to be ``no'' under the stronger constraint of $(\eps, 0)$-DP (i.e., \emph{pure} DP).
Bun, Kamath, Steinke, and Wu~\cite{BunKSW19} showed that the covering and packing numbers of a distribution class $\cH$ give sample complexity upper and lower bounds, respectively, for learning the class $\cH$.
Consequently, this immediately gives separations between learning and $(\eps, 0)$-DP learning.\footnote{The simplest natural example is the class of univariate unit-variance Gaussians with unbounded mean.}
However, they do not prove any sample complexity lower bounds for $(\eps, \delta)$-DP (i.e., \emph{approximate} DP) learning, leaving open the possibility that every learnable distribution class is privately learnable.

On the related task of PAC learning of \emph{functions}, a rich line of work shows that there exist strong separations between non-private learning and private learning, under both $(\eps, 0)$-DP~\cite{BeimelBKN14, FeldmanX15} and $(\eps, \delta)$-DP~\cite{BunNSV15, AlonLMM19, BunLM20}. 
In particular, for approximate DP, learnability is characterized by the Littlestone dimension, rather than the VC dimension as in the non-private setting.
However, given substantial differences in the setting, it is unclear whether these separations have any implications for private distribution learning.

At a July 2022 workshop at the Fields Institute, Ashtiani explicitly conjectured an affirmative answer to Question~\ref{q:learnability}: every learnable class of distributions is privately learnable~\cite{Ashtiani22}.
Indeed, as mentioned before, the community (including contributions by Ashtiani, as well as others) has designed increasingly generic algorithms for private distribution learning~\cite{AshtianiL22, TsfadiaCKMS22, AfzaliAL23}, often depending only on a non-private learner in a black-box manner.

We weakly refute Ashtiani's conjecture.
We give an explicit class of distributions which is learnable up to constant error from a constant number of samples, but is not privately learnable with the same error guarantee with any finite number of samples.\footnote{A strong refutation would require showing that the class considered is learnable with a finite number of samples for \emph{any} target error, which we do not show in this paper.
This is left as an interesting open question.}

\begin{theorem}[Informal version of Theorem~\ref{thm:main_theorem_formal}]
\label{thm:main_thm_informal}
There exists a class of distributions $\cH$ such that, for an absolute constant $c$:
\begin{enumerate}
    \item There exists an algorithm which, given $\cO\paren{1}$ samples from any distribution $\cD \in \cH$, outputs a $\widehat{\cD} \in \cH$ such that $\pr{}{\TV\paren{\widehat{\cD}, \cD} \le c} \geq 0.9$.
    \item Any $\paren{\eps, \delta}$-DP mechanism that attains the same accuracy guarantee needs an infinite number of samples.
\end{enumerate}
\end{theorem}

We use a ``trapdoor'' construction, where the class of distributions consists of mixtures over two components.
The components are entangled, in the sense that they share the same set of parameters.
The first component encodes a ``key'' that makes it possible to identify the other component.
The second component is hard to learn individually, even without privacy.
In our setting, the first component will be a binary product distribution over $\brc{0, 1}^d$, whereas the other component will be a distribution over $\brc{\pm 1, \dots, \pm d}$.
However, we stress that $d$ will not be fixed a priori, in the sense that our class will include distributions where $d$ can be any positive integer greater than $1$.\footnote{We focus on the $d \geq 2$ case because, for $d = 1$, the two components will have overlapping supports.}
The construction will be done in a way that the mixing weight will significantly favor the second component, but samples drawn from it will give very little information about the overall distribution.
Eventually, the hardness in the private setting will be a consequence of reducing from lower bounds for private mean estimation of binary product distributions (in the appropriate error metric).
We note that conceptually-similar (but technically quite different) trapdoor constructions have recently been used to show lower bounds for PAC learning~\cite{LechnerB23} and robust learnability~\cite{BenDavidBKL23}.

\paragraph{Related Work.} Gaussians are often the first class studied when considering distribution learning. They have been studied under the constraint of differential privacy starting from the work of Karwa and Vadhan on estimating univariate Gaussians~\cite{KarwaV18}, with subsequent works focusing on understanding the multivariate setting~\cite{KamathLSU19, BunS19, BiswasDKU20, LiuKKO21, AdenAliAK21, CaiWZ21, TsfadiaCKMS22, AshtianiL22, KamathMSSU22, KothariMV22, BieKS22, KamathMS22, AlabiKTVZ23, HopkinsKMN23, AsiUZ23, KamathMRSSU23}, as well as the related problem of binary product distributions~\cite{KamathLSU19, Singhal23}.
The natural generalization to learning mixtures of Gaussians has also been studied~\cite{NissimRS07, KamathSSU19, AdenAliAL21, AshtianiL22, ArbasAL23, AfzaliAL23}.
Some work focuses on estimating structured classes of distributions~\cite{DiakonikolasHS15}.
Other works study broad tools for distribution learning~\cite{BunKSW19, AdenAliAK21, AcharyaSZ21, TsfadiaCKMS22,AshtianiL22}.
See~\cite{KamathU20} for a survey of the area.

\section{Preliminaries}
\label{sec:prelim}

\textbf{General Notation.}
We denote the set of all non-zero integers by $\bZ^*$.
Additionally, given a set $S$, we define $S^i$ to be the \emph{$i$-fold Cartesian product of the set with itself}.
We use the notation $\brk{n} \coloneqq \brc{1, 2, \dots, n}$ and $\brk{a \pm R} \coloneqq \brk{a - R, a + R}$.
Also, for convenience, we will use the notations like $\paren{\bR^d}^n \equiv \bR^{n \times d}$ and $\paren{\brc{0, 1}^d}^n \equiv \brc{0, 1}^{n \times d}$.
We use $\Bern\paren{p}$ to denote a Bernoulli distribution with probability of success $p$.
Furthermore, given any set $S$, we denote the set of all distributions over that set by $\Delta\paren{S}$.
For any distribution $\cD$, $\cD^{\otimes n}$ denotes the \emph{product measure} where each \emph{marginal distribution} is $\cD$.
Thus, if we are given $n$ independent samples from $\cD$, we write $\paren{X_1, \dots, X_n} \sim \cD^{\otimes n}$.
Also, depending on the context, we may use capital Latin characters like $X$ to denote either an individual sample from a distribution or a collection of samples $X \coloneqq \paren{X_1, \dots, X_n}$.
To denote the $j$-th component of a vector, we will use a subscript (e.g., $X_j$, if the vector is $X$).
Given a pair of distributions $\cD_1, \cD_2$ over a space $\cX$, their TV-distance is defined as $\TV\paren{\cD_1, \cD_2} \coloneqq \sup\limits_{A \subseteq \cX} \abs{\cD_1\paren{A} - \cD_2\paren{A}}$.
If $\cD_1$ and $\cD_2$ are discrete, it holds that $\TV\paren{\cD_1, \cD_2} = \frac{1}{2} \sum\limits_{x \in \cX} \abs{\cD_1\paren{x} - \cD_2\paren{x}}$.
We include a few more standard facts in Appendix~\ref{sec:appendix}.

We conclude this section by introducing the definition of differential privacy and its \emph{closure under post-processing} property.

\begin{definition}[Differential Privacy (DP)~\cite{DworkMNS06}]
\label{def:dp}
A mechanism $M \colon \cX^n \to \cY$ is said to satisfy $\paren{\eps, \delta}$-differential privacy ($\paren{\eps, \delta}$-DP) if for every pair of neighboring datasets $X, X' \in \cX^n$ (i.e., datasets that differ in exactly one entry), we have:
\[
    \pr{M}{M\paren{X} \in Y} \le e^{\eps} \pr{M}{M\paren{X'} \in Y} + \delta,~~~ \forall~ Y \subseteq \cY.
\]
When $\delta = 0$, we say that $M$ satisfies $\eps$-differential privacy or pure differential privacy.
\end{definition}

\begin{lemma}[Post Processing~\cite{DworkMNS06}]
\label{lem:post_processing}
If $M \colon \cX^n \rightarrow \cY$ is $\paren{\eps, \delta}$-DP, and $P \colon \cY \to \cZ$ is any randomized function, then the algorithm $P \circ M$ is $\paren{\eps,\delta}$-DP.
\end{lemma}

\section{The Construction and Proofs}
\label{sec:main}

We define the class of distributions $\cH_{w, d} \coloneqq \brc{\cD_{w, d, p} \colon p \in \brk{0, 1}^d} \subseteq \Delta\paren{\brc{0, 1}^d \cup \brc{\pm 1, \dots, \pm d}}$, where each $\cD_{w, d, p}$ has pmf $q_{w, d, p}$ with:
\begin{equation}
    q_{w, d, p}\paren{x} \coloneqq
    \begin{cases} 
        w \prod\limits_{j \in \brk{d}} p_j^{x_j} \paren{1 - p_j}^{1 - x_j},                          & \forall x \in \brc{0, 1}^d \\
        \frac{1 - w}{d} p_1^{\frac{1 + x}{2}} \paren{1 - p_1}^{\frac{1 - x}{2}},                     & \forall x \in \brc{\pm 1} \\
        \frac{1 - w}{d} p_2^{\frac{1 + \frac{x}{2}}{2}} \paren{1 - p_2}^{\frac{1 - \frac{x}{2}}{2}}, & \forall x \in \brc{\pm 2} \\
                                                                                                     & \vdots \\
        \frac{1 - w}{d} p_d^{\frac{1 + \frac{x}{d}}{2}} \paren{1 - p_d}^{\frac{1 - \frac{x}{d}}{2}}, & \forall x \in \brc{\pm d}
    \end{cases}. \label{eq:construction}
\end{equation}
Simply put, each $\cD_{w, d, p}$ is a mixture of $d + 1$ components.
The first component has mixing weight $w$ and is a binary product distribution over $\brc{0, 1}^d$ with probability vector $p$.
Each of the remaining components has mixing weight $\frac{1 - w}{d}$ and is a binary distribution that takes the value $j$ with probability $p_j$ and the value $- j$ with probability $1 - p_j$.
Note, in particular, that the probability vector $p$ is shared for both components of the distribution.
In this context, the first component can be seen as the ``key'' to learning the distribution, because a single sample from it reveals information about the whole parameter vector, in contrast to the last $d$ components which, taken together, play the role of the ``hard distribution'', since a sample from it reveals information about only one component of the parameter vector.

Our goal will be to use $\cH_w \coloneqq \bigcup\limits_{d \geq 2} \cH_{w, d}$ as the class that will lead to the separation.
Specifically, we will show that the sample complexity of privately learning each $\cH_{w, d}$ is dimension-dependent.
As $d$ grows, the sample complexity will approach infinity.
At this point, we note that lower bounds shown for individual classes $\cH_{w, d}$ are also lower bounds for $\cH_w$ which, combined with our previous observation, implies that it is impossible to learn $\cH_w$ with a finite number of samples.

We denote our target error by $\alpha$.
The proof will focus on an instance of $\cH_w$ with $w = \frac{1}{2} \alpha$.
Specifically, focusing on the sub-class $\cH_{\frac{1}{2} \alpha, d}$ for $d \geq 2$, we first show a lower bound of $\Omega\paren{\frac{\sqrt{d}}{\log\paren{\frac{1}{\alpha}} \sqrt{\alpha} \eps}}$ for density estimation up to error $\alpha$ with probability of success $0.9$ for this class under $\paren{\eps, \delta}$-DP (Corollary~\ref{cor:private_lb_const}), and then argue that the non-private sample complexity for the same task is $\cO\paren{\frac{1}{\alpha^3}}$ (Lemma~\ref{lem:non_priv_ub}).
We conclude by formally establishing the desired separation in Theorem~\ref{thm:main_theorem_formal}.
Throughout the whole section, we will work with algorithms whose input comes from the set $\paren{\bigcup\limits_{d \geq 2} \brc{0, 1}^d} \cup \bZ^*$, while their output range is the set of all distributions over the previous domain.

We start by showing the lower bound under privacy.
Doing so involves an argument which establishes a reduction from parameter estimation for binary product distributions to density estimation for the class $\cH_{\frac{1}{2} \alpha, d}$.
Formulating the reduction first necessitates showing how a mechanism that performs density estimation for the class $\cH_{\frac{1}{2} \alpha, d}$ can be used to construct a mechanism that estimates the parameter $p$ that is associated with each distribution in this class.

\begin{lemma}
\label{lem:reduction1}
Let $d \geq 2, p \in \brk{0, 1}^d$, and $X \sim \cD_{\frac{1}{2} \alpha, d, p}^{\otimes n}$.
For $\eps, \delta \geq 0$, if $M$ is an $\paren{\eps, \delta}$-DP mechanism that takes $X$ as input and outputs a $\widehat{\cD}$ such that $\ex{X, M}{\TV\paren{\widehat{\cD}, \cD_{\frac{1}{2} \alpha, d, p}}} \le \alpha \le 1$, then it is possible to output a $\widehat{p} \in \brk{0, 1}^d$ such that $\ex{X, M}{\mnorm{1}{\widehat{p} - p}} \le 2 d \alpha$, while preserving $\paren{\eps, \delta}$-DP.
\end{lemma}

\begin{proof}
As a starting point, we note that we can extract the value $p_j$ as a function of the mass assigned to the points $\pm j$ by $\cD_{\frac{1}{2} \alpha, d, p}$.
Indeed, we have by (\ref{eq:construction}):
\begin{equation}
    q_{\frac{1}{2} \alpha, d, p}\paren{j} - q_{\frac{1}{2} \alpha, d, p}\paren{- j} = \frac{1 - \frac{1}{2} \alpha}{d} \paren{2 p_j - 1} \iff p_j = \frac{1}{2} + \frac{1}{2} \cdot \frac{d}{1 - \frac{1}{2} \alpha} \paren{q_{\frac{1}{2} \alpha, d, p}\paren{j} - q_{\frac{1}{2} \alpha, d, p}\paren{- j}}. \label{eq:p_of_q}
\end{equation}
Assume now that we have a distribution $\widehat{\cD}$ that has been outputted by $M$.
Our goal is to construct a vector $\widehat{p}$ from $\widehat{\cD}$.
To do that, we set:
\begin{equation}
    \widehat{p}_j \coloneqq \clip_{\brk{0, 1}}\paren{\frac{1}{2} + \frac{1}{2} \cdot \frac{d}{1 - \frac{1}{2} \alpha} \paren{\widehat{\cD}\paren{j} - \widehat{\cD}\paren{- j}}}, \forall j \in \brk{d}, \label{eq:p_hat}
\end{equation}
where $\clip_{\brk{0, 1}}\paren{\cdot}$ denotes the operation of clipping a number to the interval $\brk{0, 1}$ by rounding it the closest endpoint of the interval, if the number happens to fall outside it.
This operation is necessary to account for the possibility that $M$ may be \emph{improper}.

Observe now that, given a point $x \in \bR$ and a point $y \in \brk{0, 1}$, the clipping operation $\clip_{\brk{0, 1}}\paren{x}$ can never increase the distance from $y$, i.e., $\abs{\clip_{\brk{0, 1}}\paren{x} - y} \le \abs{x - y}$.
Combining this observation with (\ref{eq:p_of_q}) and (\ref{eq:p_hat}), we have:
\begin{align*}
    \abs{\widehat{p}_j - p_j}
    &\le \frac{1}{2} \cdot \frac{d}{1 - \frac{1}{2} \alpha} \abs{\paren{\widehat{\cD}\paren{j} - q_{\frac{1}{2} \alpha, d, p}\paren{j}} - \paren{\widehat{\cD}\paren{- j} - q_{\frac{1}{2} \alpha, d, p}\paren{- j}}} \\
    &\le d \paren{\abs{\widehat{\cD}\paren{j} - q_{\frac{1}{2} \alpha, d, p}\paren{j}} + \abs{\widehat{\cD}\paren{- j} - q_{\frac{1}{2} \alpha, d, p}\paren{- j}}},
\end{align*}
where we appealed to the triangle inequality and the fact that $\alpha \le 1$.

Using this, we obtain:
\[
    \mnorm{1}{\widehat{p} - p} \le d \sum\limits_{j \in \brk{d}} \paren{\abs{\widehat{\cD}\paren{j} - q_{\frac{1}{2} \alpha, d, p}\paren{j}} + \abs{\widehat{\cD}\paren{- j} - q_{\frac{1}{2} \alpha, d, p}\paren{- j}}} \le 2 d \cdot \TV\paren{\widehat{\cD}, \cD_{\frac{1}{2} \alpha, d, p}}.
\]
Taking expectation over $X$ and $M$ on the previous yields the desired accuracy guarantee, while the privacy guarantee follows from closure under post-processing (Lemma~\ref{lem:post_processing}).
\end{proof}

\begin{remark}
\label{rem:improper}
Over the course of the proof of Lemma~\ref{lem:reduction1} we highlighted the fact that our argument accounts for improper learners (rather than just proper learners).
This is an important point, since it implies that all subsequent results that build on this lemma apply to improper learners too.
\end{remark}

To complete the reduction, we need to show how, given an instance of $\ell_1$-parameter estimation for binary product distributions, it is possible to construct an instance of density estimation for the class $\cH_{\frac{1}{2} \alpha, d}$, and use a mechanism for learning $\cH_{\frac{1}{2} \alpha, d}$ as a black-box to solve the original problem instance.
This is done in the following lemma:

\begin{lemma}
\label{lem:reduction2}
For $d \geq 2$, let $P$ be a binary product distribution over $\brc{0, 1}^d$ with mean vector $p \in \brk{0, 1}^d$, and let $X \sim P^{\otimes n}$.
For $\eps, \delta \geq 0$, if any $\paren{\eps, \delta}$-DP mechanism $T \colon \brc{0, 1}^{n \times d} \to \brk{0, 1}^d$ with $\ex{X, T}{\mnorm{1}{T\paren{X} - p}} \le 2 d \alpha$ for all $p$ requires at least $n_0$ samples, the same sample complexity lower bound holds for any $\paren{\eps, \delta}$-DP mechanism $M$ that, for any $\cD_{\frac{1}{2} \alpha, d, p}$, takes $Y \sim \cD_{\frac{1}{2} \alpha, d, p}^{\otimes n}$ as input and outputs a $\widehat{\cD}$ with $\ex{Y, M}{\TV\paren{\widehat{\cD}, \cD_{\frac{1}{2} \alpha, d, p}}} \le \alpha \le 1$.
\end{lemma}

\begin{proof}
To establish our result, it suffices to show that estimating the parameter vector of $P$ can be transformed into an instance of density estimation for distributions in $\cH_{\frac{1}{2} \alpha, d}$, implying that lower bounds for the former problem also apply to the latter.
To do so, we assume we have an $\paren{\eps, \delta}$-DP mechanism $M$ that takes $Y \sim \cD_{\frac{1}{2} \alpha, d, p}^{\otimes n}$, and outputs $\widehat{\cD}$ with $\ex{Y, M}{\TV\paren{\widehat{\cD}, \cD_{\frac{1}{2} \alpha, d, p}}} \le \alpha \le 1$.
We will show how to use this mechanism to construct an $\paren{\eps, \delta}$-DP mechanism $T \colon \brc{0, 1}^{n \times d} \to \brk{0, 1}^d$ with $\ex{X, T}{\mnorm{1}{T\paren{X} - p}} \le 2 d \alpha$ for $X \sim P^{\otimes n}$.

The crux of the argument involves proving that, given a dataset $X \sim P^{\otimes n}$, it is possible to generate a dataset $Y \sim \cD_{\frac{1}{2} \alpha, d, p}^{\otimes n}$.
The mechanism $T$ will consist of this sampling step (pre-processing), and an application of $M$ over the resulting dataset.
Appealing to Lemma~\ref{lem:reduction1} suffices to establish that $T$ will have the desired accuracy guarantee, so the rest of the proof is devoted to describing the sampling process.

Given any datapoint $X_i$, we set $Y_i$ equal to it with probability $\frac{1}{2} \alpha$, or, with probability $1 - \frac{1}{2} \alpha$, we choose one of the coordinates of $X_i$ uniformly at random (say the $j$-th coordinate).
If the $j$-th coordinate of $X_i$ is equal to $1$, we set $Y_i = j$.
Otherwise, we set $Y_i=- j$.
The resulting dataset $Y$ will follow the desired distribution.
We stress that this process preserves privacy guarantees, because changing a point of $X$ can result in at most one point of $Y$ changing (conditioned on the randomness involved in the conversion of $X$ to $Y$).
\end{proof}

At this point, we recall the following result from~\cite{KamathLSU19}:

\begin{proposition}
\label{prop:fingerprinting}[Lemma~$6.2$ from~\cite{KamathLSU19}]
Let $p$ be any vector in $\brk{\frac{1}{3}, \frac{2}{3}}^d$, and let $X \coloneqq \paren{X_1, \dots, X_n}$ be a dataset consisting of $n$ independent samples from a binary product distribution $P$ over $\brc{0, 1}^d$ with mean $p$.
If $M \colon \brc{0, 1}^{n \times d} \to \brk{\frac{1}{3}, \frac{2}{3}}^d$ is an $\paren{\eps, \delta}$-DP mechanism with $\eps \in \brk{0, 1}$ and $\delta = \cO\paren{\frac{1}{n}}$ that satisfies $\ex{X, M}{\llnorm{M\paren{X} - p}^2} \le \alpha^2 \le \cO\paren{d}, \forall p \in \brk{\frac{1}{3}, \frac{2}{3}}^d$, it must hold that $n \geq \Omega\paren{\frac{d}{\alpha \eps}}$.
\end{proposition}

While phrased in terms of mechanisms with mean-squared-error guarantees, the above result also implies a bound for $\ell_1$-estimation.
The connection is described in the following lemma:

\begin{lemma}
\label{lem:private_lb_exp}
For $d \geq 2$, a sufficiently small absolute constant $\newconstant{c}$ $\useconstant{c} \in \paren{0, 1}$, and any $\alpha \le \useconstant{c}$, consider the class of distributions $\cH_{\frac{1}{2} \alpha, d}$.
Let $p \in \brk{\frac{1}{3}, \frac{2}{3}}^d$, and $X \sim \cD_{\frac{1}{2} \alpha, d, p}^{\otimes n}$.
If $M$ is an $\paren{\eps, \delta}$-DP mechanism with $\eps \in \brk{0, 1}$ and $\delta = \cO\paren{\frac{1}{n}}$ that takes $X$ as input and outputs a $\widehat{\cD}$ such that $\ex{X, M}{\TV\paren{\widehat{\cD}, \cD_{\frac{1}{2} \alpha, d, p}}} \le \alpha, \forall p \in \brk{\frac{1}{3}, \frac{2}{3}}^d$, it must hold that $n \geq \Omega\paren{\frac{\sqrt{d}}{\sqrt{\alpha} \eps}}$.
\end{lemma}

\begin{proof}
We recall the inequality $\llnorm{x}^2 \le \mnorm{\infty}{x} \mnorm{1}{x}, \forall x \in \bR^d$.
This is a consequence of \emph{H\"older's inequality}, but can also be shown in an elementary way by remarking that:
\[
    \llnorm{x}^2 = \sum\limits_{i \in \brk{d}} x_i^2 \le \max\limits_{i \in \brk{d}} \brc{\abs{x_i}} \sum\limits_{i \in \brk{d}} \abs{x_i} = \mnorm{\infty}{x} \mnorm{1}{x}.
\]
Now, let $X$ be a dataset of size $n$ that has been drawn i.i.d.\ from a binary product distribution $P$ with mean vector $p$, and let $T \colon \brc{0, 1}^{n \times d} \to \brk{0, 1}^d$ be an $\paren{\eps, \delta}$-DP mechanism with $\eps \in \brk{0, 1}, \delta = \cO\paren{\frac{1}{n}}$ that satisfies $\ex{X, T}{\mnorm{1}{T\paren{X} - p}} \le 2 d \alpha$.
We have $\mnorm{\infty}{T\paren{X} - p} \le 1$ which, by an application of the above inequality, yields $\llnorm{T\paren{X} - p}^2 \le \mnorm{1}{T\paren{X} - p}$.
This implies that $T$ satisfies the guarantee $\ex{X, T}{\llnorm{T\paren{X} - p}^2} \le 2 d \alpha$.
Additionally, if we were to project the output of $T$ to the set $\brk{\frac{1}{3}, \frac{2}{3}}^d$, this would affect neither the accuracy guarantee (since $p \in \brk{\frac{1}{3}, \frac{2}{3}}^d$), nor the privacy guarantee (due to closure under post-processing -- Lemma~\ref{lem:post_processing}).
Consequently, the lower bound of Proposition~\ref{prop:fingerprinting} applies to $T$ if we set $\alpha \to \sqrt{2 d \alpha}$.
Then, appealing to Lemma~\ref{lem:reduction2} completes the proof.
\end{proof}

The lower bound of Lemma~\ref{lem:private_lb_exp} also holds for mechanisms that achieve the accuracy guarantee $\pr{X, M}{\TV\paren{\widehat{\cD}, \cD_{\frac{1}{2} \alpha, d, p}} \le \alpha} \geq 0.9$, albeit at the cost of getting a result that's weaker by a log-factor.
The argument is sketched in the proof of Theorem~$6.1$ of~\cite{KamathLSU19}, so we point readers there and do not repeat it here.
The resulting sample complexity bound is $n \geq \Omega\paren{\frac{\sqrt{d}}{\log\paren{\frac{1}{\alpha}} \sqrt{\alpha} \eps}}$.

We summarize the above remarks in the following corollary.

\begin{corollary}
\label{cor:private_lb_const}
For $d \geq 2$, a sufficiently small absolute constant $\useconstant{c} \in \paren{0, 1}$, and any $\alpha \le \useconstant{c}$, consider the class of distributions $\cH_{\frac{1}{2} \alpha, d}$.
Let $p \in \brk{\frac{1}{3}, \frac{2}{3}}^d$, and let $X \sim \cD_{\frac{1}{2} \alpha, d, p}^{\otimes n}$.
If $M$ is an $\paren{\eps, \delta}$-DP mechanism with $\eps \in \brk{0, 1}$ and $\delta = \cO\paren{\frac{1}{n}}$ that takes $X$ as input and outputs a $\widehat{\cD}$ such that $\pr{X, M}{\TV\paren{\widehat{\cD}, \cD_{\frac{1}{2} \alpha, d, p}} \le \alpha} \geq 0.9, \forall p \in \brk{\frac{1}{3}, \frac{2}{3}}^d$, it must hold that $n \geq \Omega\paren{\frac{\sqrt{d}}{\log\paren{\frac{1}{\alpha}} \sqrt{\alpha} \eps}}$.
\end{corollary}

We now proceed to argue that it is possible to solve non-private proper density estimation with respect to the TV-distance for the class $\cH_{\frac{1}{2} \alpha, d}$ with a sample complexity that is independent of $d$.

\begin{lemma}
\label{lem:non_priv_ub}
For $d \geq 2$, there exists an algorithm $\cA$ which, for any $\alpha, \beta \in \brk{0, 1}$, given a dataset $X \sim \cD_{\frac{1}{2} \alpha, d, p}^{\otimes n}$ of size $n = \cO\paren{\frac{\log\paren{\frac{1}{\beta}}}{\alpha^3}}$, outputs a distribution $\widehat{\cD} \equiv \cD_{\frac{1}{2} \alpha, d, \widehat{p}} \in \cH_{\frac{1}{2} \alpha, d}$ such that:
\[
    \pr{X}{\TV\paren{\cD_{\frac{1}{2} \alpha, d, \widehat{p}}, \cD_{\frac{1}{2} \alpha, d, p}} \le \alpha} \geq 1 - \beta.
\]
\end{lemma}

\begin{proof}
We observe that all the distributions in the class are mixtures with components that have disjoint supports, and that the mixing weights are the same for all distributions.
As a consequence, given a pair $\widehat{p}, p \in \brk{0, 1}^d$, we have the following for the corresponding distributions:
\[   
    \TV\paren{\cD_{\frac{1}{2} \alpha, d, \widehat{p}}, \cD_{\frac{1}{2} \alpha, d, p}} = \frac{1}{2} \alpha \TV\paren{\bigotimes\limits_{j \in \brk{d}} \Bern\paren{\widehat{p}_j}, \bigotimes\limits_{j \in \brk{d}} \Bern\paren{p_j}} + \frac{1 - \frac{1}{2} \alpha}{d} \mnorm{1}{\widehat{p} - p}.
\]
By the above, to attain error $\alpha$ in TV-distance, it suffices to $\paren{1}$ estimate $\bigotimes\limits_{j \in \brk{d}} \Bern\paren{p_j}$ up to error $1$ in TV-distance, and $\paren{2}$ estimate the vector $p$ up to error $\frac{1}{2} d \alpha$ in $\ell_1$-distance.
$\paren{1}$ holds trivially, since all distributions have TV-distance at most $1$ from each other, so we focus on $\paren{2}$.

For $\paren{2}$, it holds that $\mnorm{1}{\widehat{p} - p} \le \sqrt{d} \llnorm{\widehat{p} - p}$ (Fact~\ref{fact:cs_app}), so it suffices to have a $\widehat{p}$ such that $\llnorm{\widehat{p} - p} \le \frac{1}{2} \sqrt{d} \alpha$.
Assume, now, that we are given samples drawn i.i.d.\ from a binary product distribution, and that we want to estimate its parameter vector within $\ell_2$-error $\alpha$ with probability at least $1 - \frac{1}{2} \beta$.
By Fact~\ref{fact:bin_prod_dist_est}, $\Theta\paren{\frac{d + \log\paren{\frac{1}{\beta}}}{\alpha^2}}$ samples are both necessary and sufficient for this task, with the bound being attained by taking the sample mean.
Thus, setting $\alpha \to \frac{1}{2} \sqrt{d} \alpha$ yields $\Theta\paren{\frac{d + \log\paren{\frac{1}{\beta}}}{d \alpha^2}}$, which is dominated by $\cO\paren{\frac{\log\paren{\frac{1}{\beta}}}{\alpha^2}}$.
Consequently, in order to get $\llnorm{\widehat{p} - p} \le \frac{1}{2} \sqrt{d} \alpha$ in our setting, it suffices to have $m = \cO\paren{\frac{\log\paren{\frac{1}{\beta}}}{\alpha^2}}$ samples from the first component (the binary product distribution).
For that reason, assume that, for each datapoint $X_i$ we draw from $\cD_{\frac{1}{2} \alpha, d, p}$, we have an associated random variable $Z_i \sim \Bern\paren{\frac{1}{2} \alpha}$ which becomes $1$ if $X_i$ comes from the first component.
We assume now that we have $n$ samples with $\frac{1}{2} n \alpha \geq m$.
We will show that $n = \cO\paren{\frac{\log\paren{\frac{1}{\beta}}}{\alpha^3}}$ suffices to ensure that the event $\sum\limits_{i \in \brk{n}} Z_i < m$ does not happen, except with probability at most $\frac{1}{2} \beta$.
The Hoeffding bound (Fact~\ref{fact:hoeffding}) implies that:
\[
    \pr{}{\sum\limits_{i \in \brk{n}} Z_i < m} \le \pr{}{\abs{\sum\limits_{i \in \brk{n}} Z_i - \frac{1}{2} n \alpha} \geq \frac{1}{2} n \alpha - m} \le 2 \exp\paren{- \frac{\paren{n \alpha -  2m}^2}{2 n}}.
\]
To ensure that the above is upper-bounded by $\frac{1}{2} \beta$, it suffices to have:
\[
    n \geq \frac{2 \paren{2 \alpha m + \log\paren{\frac{2}{\beta}}}}{\alpha^2} = \cO\paren{\frac{\log\paren{\frac{1}{\beta}}}{\alpha^3}}.
\]
By a union bound, the total probability of failure is upper-bounded by $\beta$, completing the proof.
\end{proof}

We are now ready to establish our main result.

\begin{theorem}
\label{thm:main_theorem_formal}
For a sufficiently small absolute constant $\useconstant{c} \in \paren{0, 1}$, we have:
\begin{enumerate}
    \item There exists an algorithm $\cA$ which, given $n = \cO\paren{1}$ samples drawn i.i.d.\ from any $\cD \in \cH_{\frac{1}{2} \useconstant{c}}$, outputs a $\widehat{\cD} \in \cH_{\frac{1}{2} \useconstant{c}}$ such that $\pr{X}{\TV\paren{\widehat{\cD}, \cD} \le \useconstant{c}} \geq 0.9$.
    \item For finite $n$, there exists no $\paren{\eps, \delta}$-DP mechanism $M$ with $\eps \in \brk{0, 1}, \delta = \cO\paren{\frac{1}{n}}$ such that, for any $\cD \in \cH_{\frac{1}{2} \useconstant{c}}$, given a dataset $X \sim \cD^{\otimes n}$, outputs a $\widehat{\cD}$ that satisfies $\pr{X, M}{\TV\paren{\widehat{\cD}, \cD} \le \useconstant{c}} \geq 0.9$.
\end{enumerate}
\end{theorem}

\begin{proof}
We will establish each part of the theorem statement separately.

Without privacy, assume we have a distribution $\cD \equiv \cD_{\frac{1}{2} \useconstant{c}, d, p} \in \cH_{\frac{1}{2} \useconstant{c}, d}$ with $p \in \brk{0, 1}^d$ as our ground truth.
Without privacy, all the algorithm $\cA$ has to do is work as we described in the proof of Lemma~\ref{lem:non_priv_ub}.
In that lemma, we showed how many samples suffice to obtain an adequate number of samples from the first component (the binary product distribution component).
Thus, the algorithm will look at the samples from that component to identify $d$, and then calculate the corresponding sample mean.
The desired guarantee then is immediate by the guarantees of that lemma.

Under privacy, we will establish our result by contradiction.
Let a distribution $\cD \equiv \cD_{\frac{1}{2} \useconstant{c}, d, p} \in \cH_{\frac{1}{2} \useconstant{c}, d}$ with $p \in \brk{\frac{1}{3}, \frac{2}{3}}^d$ be our ground truth.
Assume that, for some finite $n$, there exists an $\paren{\eps, \delta}$-DP mechanism $M$ with $\eps \in \brk{0, 1}, \delta = \cO\paren{\frac{1}{n}}$ which, given $X \sim \cD^{\otimes n}$, outputs a $\widehat{\cD}$ such that $\pr{X, M}{\TV\paren{\widehat{\cD}, \cD} \le \useconstant{c}} \geq 0.9$ for all such $\cD$.
Then, by Corollary~\ref{cor:private_lb_const}, it must be the case that $n \geq \Omega\paren{\frac{\sqrt{d}}{\eps}}$.
This must hold for every $d \geq 2$, so taking $d \to \infty$ leads to a contradiction.
\end{proof}

\section{Acknowledgments}
\label{sec:ack}

We would like to thank Thomas Steinke for helpful discussions related to the reduction from $\ell_2$-estimation to $\ell_1$-estimation, as well as the reviewers at ALT for their comments which led to the improvement of the paper's presentation.
We would further like to thank the Fields Institute for Research in Mathematical Sciences, and the organizers of the Summer 2022 Workshop on Differential Privacy and Statistical Data Analysis, where Ashtiani presented this conjecture.

MB is supported by NSF CNS-2046425 and a Sloan Research Fellowship.
GK, AM, and VS are supported by a Canada CIFAR AI Chair, an NSERC Discovery Grant, and an unrestricted gift from Google.
AM is further supported by a scholarship from the Onassis Foundation (Scholarship ID: F ZT 053-1/2023-2024).
Finally, VS is supported in part by NSF grant BCS-2218803.

\bibliographystyle{alpha}
\bibliography{biblio.bib}

@misc{Ashtiani22,
  author        = {Ashtiani, Hassan},
  title         = {Private Learning of Gaussians and their Mixtures},
  year          = {2022},
  month         = {July},
  howpublished  = {\url{https://www.youtube.com/watch?v=bmNjm0lx50I}}
}

@article{Singhal23,
  title         = {A Polynomial Time, Pure Differentially Private Estimator for Binary Product Distributions},
  author        = {Singhal, Vikrant},
  journal       = {arXiv preprint arXiv:2304.06787},
  year          = {2023}
}

@article{KamathMRSSU23,
  title         = {A Bias-Variance-Privacy Trilemma for Statistical Estimation},
  author        = {Kamath, Gautam and Mouzakis, Argyris and Regehr, Matthew and Singhal, Vikrant and Steinke, Thomas and Ullman, Jonathan},
  journal       = {arXiv preprint arXiv:2301.13334},
  year          = {2023}
}

@inproceedings{ArbasAL23,
  title         = {Polynomial Time and Private Learning of Unbounded Gaussian Mixture Models},
  author        = {Arbas, Jamil and Ashtiani, Hassan and Liaw, Christopher},
  booktitle     = {Proceedings of the 40th International Conference on Machine Learning},
  series        = {ICML '23},
  year          = {2023},
  publisher     = {JMLR, Inc.},
  pages         = {1018--1040}
}

@inproceedings{BunLM20,
  author        = {Bun, Mark and Livni, Roi and Moran, Shay},
  title         = {An Equivalence between Private Classification and Online Prediction},
  booktitle     = {Proceedings of the 61st Annual IEEE Symposium on Foundations of Computer Science},
  series        = {FOCS '20},
  year          = {2020},
  pages         = {389--402},
  publisher     = {IEEE Computer Society},
  address       = {Washington, DC, USA}
}

@inproceedings{AcharyaSZ21,
  title         = {Differentially Private Assouad, Fano, and Le Cam},
  author        = {Acharya, Jayadev and Sun, Ziteng and Zhang, Huanyu},
  booktitle     = {Proceedings of the 32nd International Conference on Algorithmic Learning Theory},
  series        = {ALT '21},
  pages         = {48--78},
  year          = {2021},
  publisher     = {JMLR, Inc.}
}

@inproceedings{AdenAliAK21,
  author        = {{Aden-Ali}, Ishaq and Ashtiani, Hassan and Kamath, Gautam},
  title         = {On the Sample Complexity of Privately Learning Unbounded High-Dimensional Gaussians},
  booktitle     = {Proceedings of the 32nd International Conference on Algorithmic Learning Theory},
  series        = {ALT '21},
  pages         = {185--216},
  year          = {2021},
  publisher     = {JMLR, Inc.}
}

@inproceedings{AdenAliAL21,
  title         = {Privately Learning Mixtures of Axis-Aligned Gaussians},
  author        = {{Aden-Ali}, Ishaq and Ashtiani, Hassan and Liaw, Christopher},
  booktitle     = {Advances in Neural Information Processing Systems 34},
  series        = {NeurIPS '21},
  year          = {2021},
  publisher     = {Curran Associates, Inc.}
}

@article{AfzaliAL23,
  title         = {Mixtures of Gaussians are Privately Learnable with a Polynomial Number of Samples},
  author        = {Afzali, Mohammad and Ashtiani, Hassan and Liaw, Christopher},
  journal       = {arXiv preprint arXiv:2309.03847},
  year          = {2023}
}

@inproceedings{AlabiKTVZ23,
  title         = {Privately estimating a {G}aussian: Efficient, robust and optimal},
  author        = {Alabi, Daniel and Kothari, Pravesh K and Tankala, Pranay and Venkat, Prayaag and Zhang, Fred},
  booktitle     = {Proceedings of the 55th Annual ACM Symposium on the Theory of Computing},
  series        = {STOC '23},
  year          = {2023},
  address       = {New York, NY, USA},
  publisher     = {ACM}
}

@inproceedings{AlonLMM19,
  author        = {Alon, Noga and Livni, Roi and Malliaris, Maryanthe and Moran, Shay},
  title         = {Private {PAC} Learning Implies Finite {L}ittlestone Dimension},
  booktitle     = {Proceedings of the 51st Annual ACM Symposium on the Theory of Computing},
  series        = {STOC '19},
  year          = {2019},
  pages         = {852--860},
  address       = {New York, NY, USA},
  publisher     = {ACM}
}

@article{AshtianiBHLMP20,
  author        = {Ashtiani, Hassan and Ben-David, Shai and Harvey, Nicholas JA and Liaw, Christopher and Mehrabian, Abbas and Plan, Yaniv},
  title         = {Near-optimal sample complexity bounds for robust learning of gaussian mixtures via compression schemes},
  journal       = {Journal of the ACM},
  year          = {2020},
  volume        = {67},
  number        = {6},
  pages         = {32:1--32:42},
  publisher     = {ACM}
}

@inproceedings{AshtianiL22,
  title         = {Private and polynomial time algorithms for learning {G}aussians and beyond},
  author        = {Ashtiani, Hassan and Liaw, Christopher},
  booktitle     = {Proceedings of the 35th Annual Conference on Learning Theory},
  series        = {COLT '22},
  year          = {2022},
  pages         = {1075--1076}
}

@article{AsiUZ23,
  title         = {From Robustness to Privacy and Back},
  author        = {Asi, Hilal and Ullman, Jonathan and Zakynthinou, Lydia},
  journal       = {arXiv preprint arXiv:2302.01855},
  year          = {2023}
}

@article{BeimelBKN14,
  author        = {Beimel, Amos and Brenner, Hai and Kasiviswanathan, Shiva Prasad and Nissim, Kobbi},
  title         = {Bounds on the Sample Complexity for Private Learning and Private Data Release},
  journal       = {Machine Learning},
  volume        = {94},
  number        = {3},
  pages         = {401--437},
  year          = {2014},
  publisher     = {Springer}
}

@inproceedings{BieKS22,
  title         = {Private Estimation with Public Data},
  author        = {Bie, Alex and Kamath, Gautam and Singhal, Vikrant},
  booktitle     = {Advances in Neural Information Processing Systems 35},
  series        = {NeurIPS '22},
  year          = {2022},
  publisher     = {Curran Associates, Inc.}
}

@inproceedings{BenDavidBKL23,
  title         = {Distribution Learnability and Robustness},
  author        = {Ben-David, Shai and Bie, Alex and Kamath, Gautam and Lechner, Tosca},
  booktitle     = {Advances in Neural Information Processing Systems 36},
  series        = {NeurIPS '23},
  year          = {2023},
  publisher     = {Curran Associates, Inc.}
}

@inproceedings{BiswasDKU20,
  title         = {CoinPress: Practical Private Mean and Covariance Estimation},
  author        = {Biswas, Sourav and Dong, Yihe and Kamath, Gautam and Ullman, Jonathan},
  booktitle     = {Advances in Neural Information Processing Systems 33},
  series        = {NeurIPS '20},
  year          = {2020},
  pages         = {14475--14485},
  publisher     = {Curran Associates, Inc.}
}

@inproceedings{BunKSW19,
  title         = {Private Hypothesis Selection},
  author        = {Bun, Mark and Kamath, Gautam and Steinke, Thomas and Wu, Zhiwei Steven},
  booktitle     = {Advances in Neural Information Processing Systems 32},
  series        = {NeurIPS '19},
  year          = {2019},
  pages         = {156--167},
  publisher     = {Curran Associates, Inc.}
}

@inproceedings{BunNSV15,
  author        = {Bun, Mark and Nissim, Kobbi and Stemmer, Uri and Vadhan, Salil},
  title         = {Differentially Private Release and Learning of Threshold Functions},
  booktitle     = {Proceedings of the 56th Annual IEEE Symposium on Foundations of Computer Science},
  series        = {FOCS '15},
  year          = {2015},
  pages         = {634--649},
  publisher     = {IEEE Computer Society},
  address       = {Washington, DC, USA}
}

@inproceedings{BunS19,
  title         = {Average-Case Averages: Private Algorithms for Smooth Sensitivity and Mean Estimation},
  author        = {Bun, Mark and Steinke, Thomas},
  booktitle     = {Advances in Neural Information Processing Systems 32},
  series        = {NeurIPS '19},
  year          = {2019},
  pages         = {181--191},
  publisher     = {Curran Associates, Inc.}
}

@article{CaiWZ21,
  author        = {Cai, T Tony and Wang, Yichen and Zhang, Linjun},
  journal       = {The Annals of Statistics},
  number        = {5},
  pages         = {2825--2850},
  publisher     = {The Institute of Mathematical Statistics},
  title         = {The cost of privacy: Optimal rates of convergence for parameter estimation with differential privacy},
  volume        = {49},
  year          = {2021}
}

@inproceedings{DaskalakisDS12b,
  author        = {Daskalakis, Constantinos and Diakonikolas, Ilias and Servedio, Rocco A.},
  title         = {Learning {P}oisson binomial distributions},
  booktitle     = {Proceedings of the 44th Annual ACM Symposium on the Theory of Computing},
  series        = {STOC '12},
  year          = {2012},
  pages         = {709--728},
  publisher     = {ACM},
  address       = {New York, NY, USA},
}

@book{DevroyeL01,
  author        = {Devroye, Luc and Lugosi, G\'abor},
  publisher     = {Springer},
  title         = {Combinatorial methods in density estimation},
  year          = {2001}
}

@inproceedings{DiakonikolasHS15,
  author        = {Diakonikolas, Ilias and Hardt, Moritz and Schmidt, Ludwig},
  title         = {Differentially Private Learning of Structured Discrete Distributions},
  booktitle     = {Advances in Neural Information Processing Systems 28},
  series        = {NIPS '15},
  year          = {2015},
  pages         = {2566--2574},
  publisher     = {Curran Associates, Inc.}
}

@inproceedings{DworkMNS06,
  author        = {Dwork, Cynthia and McSherry, Frank and Nissim, Kobbi and Smith, Adam},
  title         = {Calibrating Noise to Sensitivity in Private Data Analysis},
  booktitle     = {Proceedings of the 3rd Conference on Theory of Cryptography},
  series        = {TCC '06},
  year          = {2006},
  pages         = {265--284},
  publisher     = {Springer},
  address       = {Berlin, Heidelberg}
}

@article{FeldmanX15,
  author        = {Feldman, Vitaly and Xiao, David},
  title         = {Sample Complexity Bounds on Differentially Private Learning via Communication Complexity},
  journal       = {SIAM Journal on Computing},
  volume        = {44},
  number        = {6},
  pages         = {1740--1764},
  year          = {2015},
  publisher     = {SIAM}
}

@inproceedings{HopkinsKMN23,
  title         = {Robustness Implies Privacy in Statistical Estimation},
  author        = {Hopkins, Samuel B and Kamath, Gautam and Majid, Mahbod and Narayanan, Shyam},
  booktitle     = {Proceedings of the 55th Annual ACM Symposium on the Theory of Computing},
  series        = {STOC '23},
  year          = {2023},
  address       = {New York, NY, USA},
  publisher     = {ACM}
}

@inproceedings{KamathLSU19,
  author        = {Kamath, Gautam and Li, Jerry and Singhal, Vikrant and Ullman, Jonathan},
  title         = {Privately Learning High-Dimensional Distributions},
  booktitle     = {Proceedings of the 32nd Annual Conference on Learning Theory},
  series        = {COLT '19},
  year          = {2019},
  pages         = {1853--1902}
}

@inproceedings{KamathMS22,
  title         = {New Lower Bounds for Private Estimation and a Generalized Fingerprinting Lemma},
  author        = {Kamath, Gautam and Mouzakis, Argyris and Singhal, Vikrant},
  booktitle     = {Advances in Neural Information Processing Systems 35},
  series        = {NeurIPS '22},
  year          = {2022},
  publisher     = {Curran Associates, Inc.}
}

@inproceedings{KamathMSSU22,
  title         = {A Private and Computationally-Efficient Estimator for Unbounded Gaussians},
  author        = {Kamath, Gautam and Mouzakis, Argyris and Singhal, Vikrant and Steinke, Thomas and Ullman, Jonathan},
  booktitle     = {Proceedings of the 35th Annual Conference on Learning Theory},
  series        = {COLT '22},
  year          = {2022},
  pages         = {544--572}
}

@inproceedings{KamathSSU19,
  title         = {Differentially Private Algorithms for Learning Mixtures of Separated {G}aussians},
  author        = {Kamath, Gautam and Sheffet, Or and Singhal, Vikrant and Ullman, Jonathan},
  booktitle     = {Advances in Neural Information Processing Systems 32},
  series        = {NeurIPS '19},
  year          = {2019},
  pages         = {168--180},
  publisher     = {Curran Associates, Inc.}
}

@article{KamathU20,
  title         = {A Primer on Private Statistics},
  author        = {Kamath, Gautam and Ullman, Jonathan},
  journal       = {arXiv preprint arXiv:2005.00010},
  year          = {2020}
}

@inproceedings{KarwaV18,
  author        = {Karwa, Vishesh and Vadhan, Salil},
  title         = {Finite Sample Differentially Private Confidence Intervals},
  booktitle     = {Proceedings of the 9th Conference on Innovations in Theoretical Computer Science},
  series        = {ITCS '18},
  year          = {2018},
  pages         = {44:1--44:9},
  publisher     = {Schloss Dagstuhl--Leibniz-Zentrum fuer Informatik},
  address       = {Dagstuhl, Germany}
}

@inproceedings{KearnsMRRSS94,
  author        = {Kearns, Michael and Mansour, Yishay and Ron, Dana and Rubinfeld, Ronitt and Schapire, Robert E. and Sellie, Linda},
  title         = {On the learnability of discrete distributions},
  booktitle     = {Proceedings of the 26th Annual ACM Symposium on the Theory of Computing},
  series        = {STOC '94},
  year          = {1994},
  pages         = {273--282},
  publisher     = {ACM},
  address       = {New York, NY, USA},
}

@inproceedings{KothariMV22,
  title         = {Private Robust Estimation by Stabilizing Convex Relaxations},
  author        = {Kothari, Pravesh K and Manurangsi, Pasin and Velingker, Ameya},
  booktitle     = {Proceedings of the 35th Annual Conference on Learning Theory},
  series        = {COLT '22},
  year          = {2022},
  pages         = {723--777}
}

@article{LechnerB23,
  title         = {Impossibility of Characterizing Distribution Learning--a simple solution to a long-standing problem},
  author        = {Lechner, Tosca and Ben-David, Shai},
  journal       = {arXiv preprint arXiv:2304.08712},
  year          = {2023}
}

@inproceedings{LiuKKO21,
  title         = {Robust and Differentially Private Mean Estimation},
  author        = {Liu, Xiyang and Kong, Weihao and Kakade, Sham and Oh, Sewoong},
  booktitle     = {Advances in Neural Information Processing Systems 34},
  series        = {NeurIPS '21},
  year          = {2021},
  publisher     = {Curran Associates, Inc.}
}

@inproceedings{LiuKO22,
  title         = {Differential privacy and robust statistics in high dimensions},
  author        = {Liu, Xiyang and Kong, Weihao and Oh, Sewoong},
  booktitle     = {Proceedings of the 35th Annual Conference on Learning Theory},
  series        = {COLT '22},
  year          = {2022},
  pages         = {1167--1246}
}

@inproceedings{NissimRS07,
  author        = {Nissim, Kobbi and Raskhodnikova, Sofya and Smith, Adam},
  title         = {Smooth Sensitivity and Sampling in Private Data Analysis},
  booktitle     = {Proceedings of the 39th Annual ACM Symposium on the Theory of Computing},
  series        = {STOC '07},
  year          = {2007},
  pages         = {75--84},
  publisher     = {ACM},
  address       = {New York, NY, USA},
}

@inproceedings{TsfadiaCKMS22,
  title         = {FriendlyCore: Practical Differentially Private Aggregation},
  author        = {Tsfadia, Eliad and Cohen, Edith and Kaplan, Haim and Mansour, Yishay and Stemmer, Uri},
  booktitle     = {Proceedings of the 39th International Conference on Machine Learning},
  series        = {ICML '22},
  year          = {2022},
  publisher     = {JMLR, Inc.},
  pages         = {21828--21863}
}

\appendix
\renewcommand{\thesection}{\Alph{section}}

\section{Standard Facts}
\label{sec:appendix}

In this appendix, we include a small number of standard facts that we use in Section~\ref{sec:main}.
We start with a simple consequence of the \emph{Cauchy-Schwarz inequality}.

\begin{fact}
\label{fact:cs_app}
Let $x \in \bR^d$.
Then, we have $\mnorm{1}{x} \le \sqrt{d} \llnorm{x}$.
\end{fact}

We continue by stating the \emph{Hoeffding inequality} for i.i.d.\ Bernoulli random variables.

\begin{fact}
\label{fact:hoeffding}
Let $\paren{X_1, \dots, X_n} \sim \Bern\paren{p}^{\otimes n}$.
Defining $S_n \coloneqq \sum\limits_{i \in \brk{n}} X_i$, we have the tail bound $\pr{}{\abs{S_n - n p} \geq t} \le 2 \exp\paren{- \frac{2 t^2}{n}}$.
\end{fact}

We conclude with a folklore fact about parameter estimation for binary product distributions.

\begin{fact}
\label{fact:bin_prod_dist_est}
Let $p \in \brk{0, 1}^d$, and let $X \coloneqq \paren{X_1, \dots, X_n} \sim \paren{\Bern\paren{p_1} \times \dots \times \Bern\paren{p_d}}^{\otimes n}$.
For $\alpha \geq 0, \beta \in \brk{0, 1}$, it holds that $n = \Theta\paren{\frac{d + \log\paren{\frac{1}{\beta}}}{\alpha^2}}$ samples are necessary and sufficient for any estimator $\cA$ that takes $X$ as input and outputs a vector $\widehat{p}$ such that $\pr{\cA, X}{\llnorm{\widehat{p} - p} \le \alpha} \geq 1 - \beta$ for any $p \in \brk{0, 1}^d$.
Moreover, this bound is attained by the empirical mean $\widehat{p} \coloneqq \frac{1}{n} \sum\limits_{i \in \brk{n}} X_i$.
\end{fact}

\end{document}